\documentclass[pdflatex,sn-basic]{sn-jnl}

\jyear{2021}%

\theoremstyle{thmstyleone}%
\newtheorem{theorem}{Theorem}
\newtheorem{lemma}{Lemma}
\newtheorem{corollary}{Corollary}

\theoremstyle{thmstyletwo}%
\newtheorem{remark}{Remark}%

\theoremstyle{thmstylethree}%
\newtheorem{definition}{Definition}%

\begin{document}

\title[Fundamental Portfolio Outperforms the Market Portfolio]{Fundamental Portfolio Outperforms the Market Portfolio}


\author[1]{\fnm{Hayden} \sur{Brown} ORCID: 0000-0002-2975-2711}\email{haydenb@nevada.unr.edu}

\affil[1]{\orgdiv{Department of Mathematics and Statistics}, \orgname{University of Nevada, Reno}, \orgaddress{1664 \street{N. Virginia Street}, \city{Reno}, \postcode{89557}, \state{Nevada}, \country{USA}}}


\abstract{There is substantial empirical evidence showing the fundamental portfolio outperforming the market portfolio. Here a theoretical foundation is laid that supports this empirical research. Assuming stock prices revert around fundamental prices with sufficient strength and symmetry, the fundamental portfolio outperforms the market portfolio in expectation. If reversion toward the fundamental price is not sufficiently strong, then the fundamental portfolio underperforms the market portfolio in expectation.}




\maketitle

\section{Introduction}
Value and growth stocks are typically differentiated based on the ratios of book-to-market equity (B/M), earnings to price (E/P), cash flow to price (C/P) and dividend to price (D/P). Investment managers associate higher ratios with value stocks and lower ratios with growth stocks. Historically, value stocks have outperformed growth stocks. Since large market portfolios like the S\&P 500 are a blend of value and growth stocks, this is encouraging evidence that portfolios which overweight value stocks can beat the market portfolio in the long run. 

One simplistic explanation for this outperformance by value stocks uses E/P. The idea is that E/P has a constant trend in the market. Stock E/Ps fluctuate above and below this trend, with a stronger tendency to revert toward the trend than away. Value stocks lie above the trend, and growth stocks lie below. Assuming constant earnings, this reversion will eventually make the price of value stocks go up and the price of growth stocks go down. Of course, earnings are not constant, and the effect of the reversion on price is more complicated. 

Here, a \textit{fundamental portfolio} is constructed that overweights value stocks and underweights growth stocks. Assuming that stock prices revert around a fundamental price, conditions are given such that the fundamental portfolio outperforms the market portfolio in expectation, not counting dividends. Surprisingly, outperformance in expectation is not guaranteed, and the reversion strength must be sufficiently strong. 

\subsection{Literature Review}
Value stocks have outperformed growth stocks on an international level. From 1974 to 1994, value stocks provided substantially higher returns than growth stocks in the United States and twelve major EAFE (Europe, Australia, and the Far East) countries \cite{fama1998value}. In the US, this outperformance is evidenced beyond just these 20 years. Value stocks provided higher returns in the US from 1963 to 1990 \cite{fama1992cross}, 1986 to 2002 \cite{chan2004value}, and 1999 to 2014 \cite{an2017value}.

Outperformance of the market portfolio is a major topic because of its implication on the existence of arbitrage. If a portfolio outperforms the market portfolio with probability 1, then arbitrage is achieved by shorting the market portfolio and using the short position to finance a long position in the other portfolio. The Capital Asset Pricing Model (CAPM) does not support outperformance with probability 1, but it does support outperformance in expectation \cite{sharpe1964capital}. Using continuous rebalancing and stock price processes which are adapted to Brownian motion, outperformance with probability 1 is possible, provided a sufficient amount of time has passed \cite{fernholz2002stochastic}. However, this outperformance does not factor in the discrete-in-time distribution of dividends. Using a similar framework as \cite{fernholz2002stochastic}, conditions are given that lead to outperformance of the market portfolio with a given probability \cite{bayraktar2012outperforming}. In reality, rebalancing is discrete in time, so it is worth investigating outperformance in this setting.

The concept of a fundamental portfolio is detailed in \cite{arnott2005fundamental}. Instead of weighting companies based on capitalization, like in the S\&P 500, companies are weighted based on some fundamental metric. Examples of fundamental metrics include gross revenue, equity book value, gross sales, gross dividends, cash flow, and total employment. In the period 1962 to 2004, such fundamentally weighted portfolios outperformed the S\&P 500. After applying a bootstrap procedure to the period 1982 to 2008, there is evidence of outperformance by fundamental portfolios on a global level, but not on a country-specific level \cite{walkshausl2010fundamental}. Some fundamental index funds that have been launched on the US market are listed in \cite{amenc2008comparison}, along with an analysis of their performance. Here, it is assumed that a given stock price reverts around a fundamental price. This fundamental price can be constructed using the abovementioned fundamental metrics, or any other relevant metrics. What matters is that the necessary reversion conditions are met. 

The fundamental portfolio has been investigated thoroughly from an emperical standpoint; a summary can be found in \cite{hsu2011fundamental}. Mean reversion of various stock-related metrics has also received significant attention. There is evidence of mean reversion in E/P for the S\&P 500 \cite{becker2012empirical}. Furthermore, individual company E/Ps tend to revert toward the industry norm over time \cite{bajaj2005mean}. In general, there is evidence supporting mean reversion in stock prices for emerging markets \cite{chaudhuri2003mean} and developed markets \cite{balvers2000mean}. Theoretical results concerning the fundamental portfolio have not been pursued with the same vigor. The connection between this mean reversion and the fundamental portfolio's performance versus the market portfolio remains unexplored. 

\begin{figure}[h] 
\begin{center}
  \includegraphics{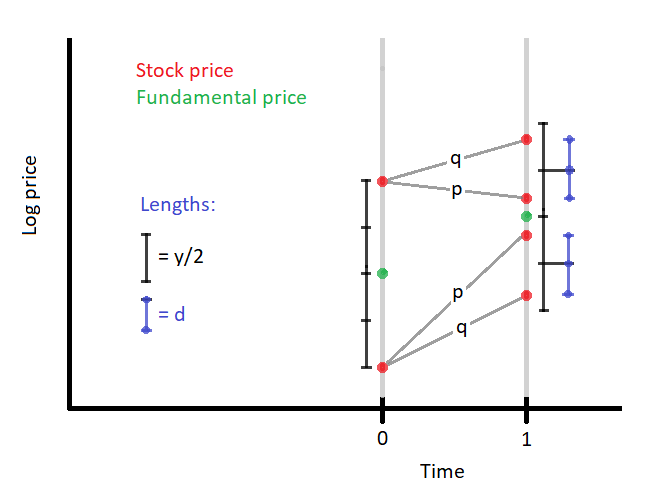}
\end{center}
  \caption{It is assumed that fundamental prices are deterministic and $y,d\geq0$. The conditions of Theorem \ref{T1} involve only the random, reverting processes added to the logarithm of the fundamental price. Each letter $p$ and $q$ has a line going through it indicating a realization of stock prices for times 0 and 1. The letters $p$ and $q$ indicate the probability of each realization. Theorem \ref{T1} requires $1\geq p\geq q\geq 0$. The symmetry shown in the probabilities is also required.
 }
  \label{fig:T1}
\end{figure}

\begin{figure}[h] 
\begin{center}
  \includegraphics{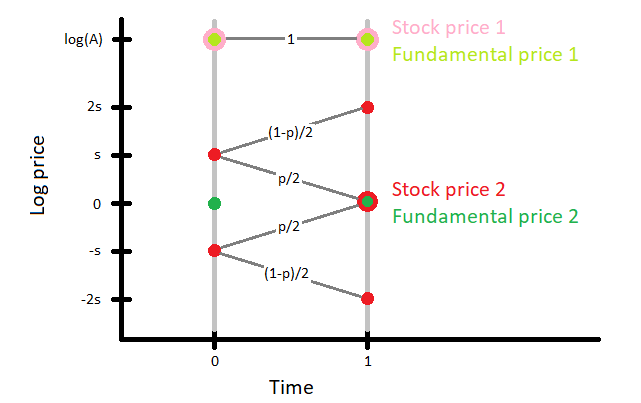}
\end{center}
  \caption{The possible realizations for the price of two stocks at times 0 and 1 are illustrated. According to Theorem \ref{Tcounter}, there exist $s>0$, $p\in(\frac{1}{2},1]$ and $A>0$ such that the fundamental portfolio is expected to underperform the market portfolio. Note that rebalancing occurs at time 0, and the underperformance in expectation occurs at time 1. Fundamental prices are assumed to be deterministic. Each non-vertical gray line connects a stock price at time 0 with a stock price at time 1, indicating a possible realization of the stock price at times 0 and 1. The expression in the middle of a particular gray line indicates the probability of that realization. 
 }
  \label{fig:Tcounter}
\end{figure}

\subsection{Results}
When stock prices revert around fundamental prices according to some general conditions on the symmetry and strength of the reversion, the fundamental portfolio outperforms the market portfolio in expectation. An upside of this result is that it holds for discrete-time rebalancing. A downside of this result is that it does not factor in dividends. More specifically, the outperformance only applies to the value of a portfolio without counting or reinvesting dividends. One other downside is that the result requires each stock price to have an independent reverting process. However, the result can apply to stock prices with dependent reverting processes, provided the linearity condition outlined at the end of Section \ref{prelim} is satisfied.

Theorem \ref{T1} provides the first set of conditions causing the fundamental portfolio to outperform the market portfolio in expectation. The conditions are illustrated in Figure \ref{fig:T1}. This result also addresses semi-fundamental portfolios, which are a hybrid of the fundamental and market portfolios. Theorem \ref{T3} provides another set of conditions leading to the same result as Theorem \ref{T1}. The advantage of Theorem \ref{T3} is that its conditions use conditional probability, which can be easier to work with in practice. Three corollaries show how Theorem \ref{T3} is applied when the reverting processes are Ornstein-Uhlenbeck processes, order 1 autoregressive processes and Markov chains. Theorem \ref{T2} relaxes the condition on reversion strength of Theorem \ref{T1}, at the cost of some additional boundary conditions on the change in stock price and fundamental price between rebalancing times.

Theorem \ref{Tcounter} shows the importance of the reversion strength condition given in Theorems \ref{T1}, \ref{T3} and \ref{T2}. Even when the reversion makes it more likely for the stock price to move toward the fundamental price than away, if this gravitation toward the fundamental price is not strong enough, the fundamental portfolio will underperform the market portfolio in expectation. The stock price processes used in Theorem \ref{Tcounter} are illustrated in Figure \ref{fig:Tcounter}.

\subsection{Organization}
Section \ref{prelim} provides basic definitions and lemmas in order to approach the fundamental versus market portfolio problem. In particular, a measure theoretic framework is constructed. Portfolios are formally defined, along with the value of an investment in a portfolio over time. Section \ref{mainresults} provides the main results. Appendix \ref{appendix} contains all proofs and some additional lemmas. Section \ref{conclusion} provides closing remarks and ideas for future research.

\section{Preliminaries}\label{prelim}
First, an $n$-dimensional stochastic process is constructed to describe the randomness in stock prices.

\begin{definition}
Take $n$ probability measure spaces $(\boldsymbol{\Omega},\mathcal{A},\mathbb{P})_i=(\Omega_i,\mathcal{A}_i,\mathbb{P}_i)$, $i=1,...,n$, and denote their product measure space in the standard sense as 
\begin{equation*}
(\boldsymbol{\Omega},\mathcal{A},\mathbb{P})=(\times_{i=1}^n\Omega_i,\otimes_{i=1}^n\mathcal{A}_i,\otimes_{i=1}^n\mathbb{P}_i).
\end{equation*}
Define the $n$-dimensional stochastic process $\mathbf{Y}:\boldsymbol\Omega\times\mathbb{R}\to\mathbb{R}$ such that $\mathbf{Y}(\boldsymbol\omega,t)=(Y_1(\omega_1,t),\ ...,\ Y_n(\omega_n,t))$, and for all $t\in\mathbb{R}$, $\mathbf{Y}(\boldsymbol\omega,t)$ is $\mathcal{A}$-measurable. In addition, require $\mathbb{E}\vert  Y_i(\boldsymbol\omega,t)\vert $ and $\mathbb{E}\exp Y_i(\boldsymbol\omega_i,t)$ to be finite for all $i\in\{1,...,n\}$ and $t\in\mathbb{R}$. 
\end{definition}

Note that $\mathbb{E}$ denotes expectation with respect to $\boldsymbol\Omega$. Furthermore, this construction of $\mathbf{Y}$ makes the $Y_i$ independent. Next, the fundamental price and stock price is defined.

\begin{definition}
For $i=1,..,n$, define the function $F_i:\mathbb{R}\to(0,\infty)$. $F_i(t)$ indicates the fundamental price of stock $i$ at time $t$. The $i$th stock price $X_i:\boldsymbol\Omega\times \mathbb{R}\to(0,\infty)$ is defined by $X_i(\boldsymbol\omega,t)=F_i(t)\exp(Y_i(\omega_i,t))$. Denote $\mathbf{F}(t)=(F_1,...,F_n)(t)$ and $\mathbf{X}(\boldsymbol\omega,t)=(X_1,...,X_n)(\boldsymbol\omega,t)$.
\end{definition}

Note that the fundamental prices $F_i(t)$ are deterministic, whereas the stock prices $X_i(\boldsymbol\omega,t)$ are random. Furthermore, $\exp :\mathbb{R}\to(0,\infty)$ is Borel measurable, so $X_i(\boldsymbol\omega,t)$ is $\mathcal{A}$-measurable for each $t\in\mathbb{R}$.

To simplify notation, functions that take inputs $(\boldsymbol\omega,t)\subset\boldsymbol\Omega\times\mathbb{R}$ or $(\omega_i,t)\subset\Omega_i\times\mathbb{R}$ may be abbreviated with the single input $(t)$. Next, the discrete rebalancing scheme is defined.

\begin{definition}
Fix $\{t_k\}_{k=0}$ as an increasing sequence in $\mathbb{R}$ with at least two elements. Investment is begun at $t_0$, and rebalancing occurs at each $t_k$. No short sales are permitted. Use $T$ to denote the set of elements in the sequence $\{t_k\}_{k=0}$. The portfolio weights at each rebalancing time are given by the portfolio function $\boldsymbol\pi:T\to(0,1)^n$, where $\boldsymbol\pi(t_k)=(\pi_i(t_k),\ ...,\ \pi_n(t_k))$, $\sum_{i=1}^n\pi_i(t_k)=1$ and $\pi_i(t_k)$ indicates the weight for stock $i$ at time $t_k$. 
\label{d:rebalance}
\end{definition}

Next, the value function, $V_{\pi}$, is defined for a particular portfolio function $\boldsymbol\pi$. It tracks the value of an investment, not counting dividends. There is no reinvestment of dividends in this model.

\begin{definition}
Use $V_\pi:\boldsymbol\Omega\times T\to[0,\infty)$ to denote the value of an investment at each rebalancing time, according to portfolio $\boldsymbol\pi$. Then
\begin{equation}
V_\pi(t_{k+1})=V_\pi(t_k)\sum_{i=1}^n\pi_i(t_k)\frac{X_i(t_{k+1})}{X_i(t_k)}.
\label{valueOriginal}
\end{equation}
Assume that $V_\pi(t_0)=1$.
\label{d:V}
\end{definition}

Lemma \ref{l:logV} describes the change in log-value between rebalancing times $t_k$ and $t_{k+1}$. Lemma \ref{l:dlogV} describes the change in log-value-difference between two portfolios and rebalancing times $t_k$ and $t_{k+1}$. 

\begin{lemma}
\eqref{valueOriginal} is written in logorithmic form, using $\Delta_k f$ as shorthand for $f(t_{k+1})-f(t_k)$.
\begin{equation}
\Delta_k\log V_\pi=\log\sum_{i=1}^n\pi_i(t_k)\exp\Delta_k\log X_i.
\label{valueLog}
\end{equation}
\label{l:logV}
\end{lemma}

\begin{lemma}
Given two portfolios $\boldsymbol\pi$ and $\boldsymbol\eta$, 
\begin{equation*}
\Delta_k\log\frac{V_\pi}{V_\eta}=\log\frac{\sum_{i=1}^n\pi_i(t_k)\exp\Delta_k\log X_i}{\sum_{i=1}^n\eta_i(t_k)\exp\Delta_k\log X_i}.
\end{equation*}
Moreover, for each $t_k\in T$ with $k>0$, $\log\frac{V_\pi(\boldsymbol\omega,t_k)}{V_\eta(\boldsymbol\omega,t_k)}$ is $\mathcal{A}$-measurable and integrable with respect to $(\boldsymbol\Omega,\mathcal{A},\mathbb{P})$. 
\label{l:dlogV}
\end{lemma}

Next, the fundamental and market portfolios are defined. $\boldsymbol\pi^m$ denotes the portfolio which uses the fundamental price as weights for the first $m$ stocks and the stock price as weights for the remaining stocks. So $\boldsymbol\pi^0$ is the market portfolio and $\boldsymbol\pi^n$ is the fundamental portfolio. When $0<m<n$, the portfolio $\boldsymbol\pi^m$ is a hybrid of the market and fundamental portfolios. 
\begin{definition}
Set
\begin{equation*}
\lambda_i^m(t_k)=\begin{cases}F_i(t_k)& i\leq m\\X_i(t_k)& \text{otherwise}\end{cases},\quad\pi_i^{m}(t_k)=\frac{\lambda_i^m(t_k)}{\sum_{j=1}^n\lambda_j^m(t_k)},\quad\boldsymbol\pi^m(t_k)=(\pi_1^m,...,\pi_n^m)(t_k),
\end{equation*}
where $t_k\in T$, $m=0,...,n$ and $i=1,...,n$. 
\end{definition}

Next, additional measure theory notation is introduced. $\mu_k^m$ measures the probability that $[Y_{m}(t_k),\Delta_kY_{m}]$ lies in some 2-dimensional Borel set. $u_k^m$ measures the probability that $[Y_{m}(t_k)]$ lies in some 1-dimensional Borel set. $\nu_k^m(y,\cdot)$ measures the conditional probability that $[Y_{m}(t_k),\Delta_kY_{m}]$ lies in some 2-dimensional Borel set, given $Y_{m}(t_k)=y$.
\begin{definition}
Let $\mathcal{R}$ be the set of rectangles in $\mathbb{R}^2$. $R\in\mathcal{R}$ implies $R=I_1\times I_2$, where $I_1$ and $I_2$ are intervals in $\mathbb{R}$. Let $\mathcal{B}=\sigma(\mathcal{R})$ denote the Borel sigma algebra for $\mathbb{R}^2$. Let $\mathcal{B}_1$ denote the Borel sigma algebra for $\mathbb{R}$. For each $t_{k+1}\in T$, $k\geq0$ and $m\in\{1,...,n\}$, define the functions
\begin{equation*}
\begin{split}
&\mu^{m}_k:\ \mathcal{B}\to[0,1],\quad \mu_k^{m}(B)=\mathbb{P}_{m}([Y_{m}(t_k),\Delta_kY_{m}]^{-1}(B))\\
&u^m_k:\ \mathcal{B}_1\to[0,1],\quad u_k^m(B)=\mathbb{P}_{m}([Y_{m}(t_k)]^{-1}(B))\\
&\nu_k^m:\ \mathbb{R}\times\mathcal{B}\to[0,1],\quad \mu_k^m(B)=\int_{\mathbb{R}}\nu_k^m(y,B)du_k^m(y).
\end{split}
\end{equation*}
\end{definition}

Intervals in $\mathbb{R}$ are sets of the form $(a,a)$, $[a,a]$, $(a,b]$, $[a,b)$, $(a,b)$, $[a,b]$, $(-\infty,a)$, $(-\infty,a]$, $(a,\infty)$, $[a,\infty)$, $(-\infty,\infty)$, where $a,b\in\mathbb{R}$. The next definition is a miscellaneous collection meant to simplify notation.
\begin{definition}
Let $A\subset\mathbb{R}^2$. Given a function $h:\mathbb{R}^2\to\mathbb{R}^2$, let $h(A)=\{h(x,y):\ (x,y)\in A\}$. Given $\mathcal{C}$, a collection of subsets of $\mathbb{R}^2$, let $\mathcal{C}_A=\{A\cap C:\ C\in\mathcal{C}\}$. Use the conventions $-A=\{(-x,-y):\ (x,y)\in A\}$, $A'=\{(x,-x-y):\ (x,y)\in A\}$ and $\bar{A}$ to indicate the closure of $A$. Define the sets
\begin{equation*}
\begin{split}
&R_1=\{(y,d_y)\in\mathbb{R}^2:\ y>0,\  d_y\geq -\frac{1}{2}y\}\\
&R_2=\{(y,d_y)\in\mathbb{R}^2:\ y>0,\  d_y> -\frac{1}{2}y\}.
\end{split}
\end{equation*}
\end{definition}

In this construction, the reverting processes $Y_i(t)$ are independent across $i$. Suppose the actual stock prices are given by $\mathbf{X}^*(\boldsymbol\omega,t)=\mathbf{A}\mathbf{X}(\boldsymbol\omega,t)$, where $\mathbf{A}$ is a fixed $n\times n$ real-valued, non-singular matrix. Note that $\mathbf{X}^*$ and $\mathbf{X}$ are both column vectors with $n$ elements each. Then it is possible to write $\mathbf{X}(\boldsymbol\omega,t)=\mathbf{A}^{-1}\mathbf{X}^*(\boldsymbol\omega,t)$, allowing the theory to be applied here as well. However, shorting may be required due to negative values in $\mathbf{A}^{-1}$. Furthermore, the results compare the fundamental and market portfolios with respect to $\mathbf{X}(\boldsymbol\omega,t)$. When the market portfolio with respect to $\mathbf{X}(\boldsymbol\omega,t)$ is expressed in terms of $\mathbf{X}^*(\boldsymbol\omega,t)$, the weights likely will not resemble the weights of the market portfolio with respect to $\mathbf{X}^*(\boldsymbol\omega,t)$.

\section{Main Results}\label{mainresults}
Theorem \ref{T1} provides the first set of measure theoretic conditions causing the fundamental portfolio to outperform the market portfolio in expectation. $(i)$ is the symmetry condition and $(ii)$ is the reversion strength condition. 

Theorem \ref{T3} provides a more practical set of conditions that give the same result as Theorem \ref{T1}. $(a)$ and $(b)$ are the analogues of $(i)$ and $(ii)$ from Theorem \ref{T1}, using conditional probability. Corollaries \ref{c:ou}, \ref{corar} and \ref{corMarkov} are applications of Theorem \ref{T3}.

Corollary \ref{c:ou} covers the case where the reverting processes are Ornstein-Uhlenbeck processes. Here, the fundamental portfolio outperforms the market portfolio in expectation when rebalancing times are sufficiently far apart.

Corollary \ref{corar} covers the case where the reverting processes are order 1 autoregressive processes. Corollary \ref{corMarkov} covers the case where the reverting processes are Markov chains. 

Theorem \ref{T2} relaxes $(ii)$ of Theorem \ref{T1}. The cost of this relaxed reversion strength condition is two boundary conditions. The boundary conditions control the change in fundamental price and stock price between rebalancing times.

Theorem \ref{Tcounter} is a counterexample to the idea that a fundamental portfolio always outperforms a market portfolio in expectation when there is symmetrical reversion toward at least one fundamental. The goal is to emphasize the importance of conditions in Theorems \ref{T1}, \ref{T3} and \ref{T2}. If those conditions are not satisfied, an investor following a fundamental portfolio could underperform the market portfolio in expectation.

\begin{theorem}
Fix $k\geq 0$ such that $t_{k+1}\in T$, and $m_1,m_2\in\{1,...,n\}$ such that $m_1\leq m_2$. Suppose the following conditions hold for all $m\in\{m_1,...,m_2\}$.
\begin{enumerate}[(i)]
\item $\mu^{m}_k(R)=\mu^{m}_k(-R)$ for all $R\in\mathcal{R}$.
\item $\mu^{m}_k(R)\leq \mu^{m}_k(R')$ for all $R\in\mathcal{R}_{R_2}$.
\end{enumerate}
Let $\odot\in\{>,\geq\}$. Then $\sum_{m=m_1}^{m_2}\mu^{m}_k(R_1)\odot 0$ implies
\begin{equation*}
\mathbb{E}\log\frac{V_{\pi^{m_2}}(t_{k+1})}{V_{\pi^{m_1-1}}(t_{k+1})}\odot \mathbb{E}\log\frac{V_{\pi^{m_2}}(t_{k})}{V_{\pi^{m_1-1}}(t_{k})}.
\end{equation*}
\label{T1} 
\end{theorem}

\begin{theorem}
Fix $K>0$ such that $t_{K+1}\in T$, and $m_1,m_2\in\{1,...,n\}$ such that $m_1\leq m_2$. Suppose the following conditions hold for all $m\in\{m_1,...,m_2\}$ and $k\leq K$.
\begin{enumerate}[(a)]
\item $\nu^{m}_k(y,R)=\nu^{m}_k(-y,-R)$ for all $R\in\mathcal{R}$, $u_k^m$-almost everywhere.
\item $\nu^{m}_k(y,R)\leq \nu^{m}_k(y,R')$ for all $R\in\mathcal{R}_{R_2}$, $u_k^m$-almost everywhere.
\item $u_0^m(I)=u_0^m(\{-y:\ y\in I\})$ for all intervals $I\subset\mathbb{R}$.
\end{enumerate}
In addition, suppose there exists $m\in\{m_1,...,m_2\}$ such that 
\begin{equation*}
\begin{split}
&\nu_{K-1}^m(y,\{(z,-z):\ z\in\mathbb{R}\})<1\ \ u_{K-1}^m\text{-almost everywhere},\\
&\nu_K^m(y,R_1)+\mathcal{X}_{(-\infty,0]}(y)>0\ \ u_K^m\text{-almost everywhere}.
\end{split}
\end{equation*} 
Then 
\begin{equation*}
\mathbb{E}\log\frac{V_{\pi^{m_2}}(t_{K+1})}{V_{\pi^{m_1-1}}(t_{K+1})}>\mathbb{E}\log\frac{V_{\pi^{m_2}}(t_{K})}{V_{\pi^{m_1-1}}(t_{K})}.
\end{equation*}
\label{T3} 
\end{theorem}

\begin{corollary}
Fix $K>0$ such that $t_{K+1}\in T$, and $m_1,m_2\in\{1,...,n\}$ such that $m_1\leq m_2$. For $m\in\{m_1,...,m_2\}$, let $Y_m$ be an Ornstein-Uhlenbeck process:
\begin{equation*}
dY_m(t)=-\theta_mY_m(t)dt+\sigma_mdW_m(t),\quad \theta_m,\sigma_m>0,
\end{equation*}
where $Y_m(0)$ is a non-trivial random variable with symmetric distribution about $0$. In addition, require that the $W_m(t)$ are independent Wiener processes. If $t_{k+1}-t_k\geq\frac{\ln 2}{\min\{\theta_1,...,\theta_M\}}$ for all $k\leq K$, then
\begin{equation*}
\mathbb{E}\log\frac{V_{\pi^{m_2}}(t_{K+1})}{V_{\pi^{m_1-1}}(t_{K+1})}>\mathbb{E}\log\frac{V_{\pi^{m_2}}(t_{K})}{V_{\pi^{m_1-1}}(t_{K})}.
\end{equation*}
\label{c:ou}
\end{corollary}

\begin{corollary}
Fix $t_k=k$ for all $k\in\mathbb{N}\cup\{0\}$, and $m_1,m_2\in\{1,...,n\}$ such that $m_1\leq m_2$. For $m\in\{m_1,...,m_2\}$, let $Y_m$ be an order 1 autoregressive process:
\begin{equation*}
Y_m(k+1)=\theta_mY_m(k)+Z_m(k+1),\quad \theta_m\leq\frac{1}{2},
\end{equation*}
where $Y_m(0),\ Z_m(1),\ Z_m(2),\ Z_m(3),\ ...$ are independent, non-trivial random variables with symmetric distribution about $0$. Further require $\mathbb{P}(Z_m(k)>a)>0$ for all $k\in\mathbb{N}$ and $a\in\mathbb{R}$. Then for all $k>0$,
\begin{equation*}
\mathbb{E}\log\frac{V_{\pi^{m_2}}(t_{k+1})}{V_{\pi^{m_1-1}}(t_{k+1})}>\mathbb{E}\log\frac{V_{\pi^{m_2}}(t_{k})}{V_{\pi^{m_1-1}}(t_{k})}.
\end{equation*}
\label{corar}
\end{corollary}

\begin{remark}
A special case of Corollary \ref{corar} is when $\theta_m=0$. In this situation, $Y_m$ is just white noise. So even when fluctuations around the fundamentals follow white noise, a semi-fundamental portfolio is expected to outperform a market portfolio.  
\end{remark}

\begin{corollary}
Fix $t_k=k$ for all $k\in\mathbb{N}\cup\{0\}$, and $m_1,m_2\in\{1,...,n\}$ such that $m_1\leq m_2$. For $m\in\{m_1,...,m_2\}$, let $Y_m$ be a discrete-time Markov chain on $\{t_k\}_k$ with state space $S:=\{ks:\ k\in\mathbb{Z}\}$ and 1-step transition probabilities from $t_k$ to $t_{k+1}$, denoted $P_m(y_1,y_2)$, that satisfy:
\begin{enumerate}[(i)]
\item $P_m(k_1s,k_2s)=P_m(-k_1s,-k_2s)$ for all $k_1,k_2\in\mathbb{Z}$,
\item $P_m(k_1s,(k_1+k_2)s)\leq P_m(k_1s,-k_2s)$ for all $k_1\in\mathbb{N}$ and $k_2\in\mathbb{Z}$ such that $k_2>-\frac{1}{2}k_1$,
\item $Y_m(0)$ maps into $S$ and has non-trivial, symmetrical distribution about $0$,
\item $P_m(k_1s,0)<1$ for all $k_1\in\mathbb{Z}$,
\item $\sum_{k_2\geq -\frac{1}{2}k_1}P_m(k_1s,(k_1+k_2)s)>0$ for all $k_1\in\mathbb{N}$.
\end{enumerate}
Then for all $k>0$,
\begin{equation*}
\mathbb{E}\log\frac{V_{\pi^{m_2}}(t_{k+1})}{V_{\pi^{m_1-1}}(t_{k+1})}>\mathbb{E}\log\frac{V_{\pi^{m_2}}(t_{k})}{V_{\pi^{m_1-1}}(t_{k})}.
\end{equation*}
\label{corMarkov} 
\end{corollary}

\begin{theorem}
Fix $k\geq 0$ such that $t_{k+1}\in T$, $\delta_1,\delta_2>0$ and $m_1,m_2\in\{1,...,n\}$ such that $m_1\leq m_2$. Suppose the following conditions hold for all $m\in\{m_1,...,m_2\}$. There is a function $r^{m}_k:\ \{(y,d_y)\in\mathbb{R}^2:\ y>0,\  d_y>-\frac{1}{2}y\}\to[0,1)$ such that
\begin{equation*}
\frac{1}{2}\left[1-\exp(-2\delta_1-\delta_2)\frac{\exp y+\exp(-y)-2}{\exp(y+d_y)+\exp(-y-d_y)-\exp d_y-\exp(-d_y)}\right]<r_k^{m}(y,d_y).
\end{equation*}
And
\begin{enumerate}[(i)]
\item $\mu^{m}_k(R)=\mu^{m}_k(-R)$ for all $R\in\mathcal{R}$.
\item $\mu^{m}_k(R')\geq\int_{R}\frac{r_k^m}{1-r_k^m}d\mu_k^m$ for all $R\in\mathcal{R}_{R_2}$.
\item $-\delta_1\leq \Delta_k\log F_{m}\leq \delta_1$.
\item $-\delta_2\leq \Delta_kY_{m}\leq\delta_2$.
\end{enumerate}
Then $\sum_{m=m_1}^{m_2}\mu^{m}_k(R_1)>0$ implies
\begin{equation*}
\mathbb{E}\log\frac{V_{\pi^{m_2}}(t_{k+1})}{V_{\pi^{m_1-1}}(t_{k+1})}>\mathbb{E}\log\frac{V_{\pi^{m_2}}(t_{k})}{V_{\pi^{m_1-1}}(t_{k})}.
\end{equation*}
\label{T2} 
\end{theorem}

\begin{remark}
It is tempting to try to modify Theorem \ref{T2} so the result is reversed. However, this is futile because it requires reversing the inequality in \eqref{dk}. It is easy to see that there is a neighborhood $d_y\in(-\frac{1}{2}y,-\frac{1}{2}y+\epsilon)$ on which \eqref{dk} with the reversed inequality is not satisfied.
\end{remark}

\begin{theorem}
Fix $n=2$, $T=\{0,1\}$, $s>0$, $F_1(0)=F_1(1)=1$ and $Y_2(0)=Y_2(1)=0$. There exists $A>0$ and $Y_1$ with $\mathbb{P}(Y_1(0)=s)=\mathbb{P}(Y_1(0)=-s)=\frac{1}{2}$ and 1-step transition probabilities from time $0$ to time $1$, denoted $M(y_1,y_2)$, satisfying:
\begin{enumerate}[(i)]
\item $M(s,(1\pm1)s)=M(-s,-(1\pm1)s)$,
\item $M(s,2s)<M(s,0)$,
\item $M(s,2s)+M(s,0)=1$,
\end{enumerate}
such that if $F_2(0)=F_2(1)=A$, then
\begin{equation*}
\mathbb{E}\log\frac{V_{\pi^{2}}(1)}{V_{\pi^0}(1)}<\mathbb{E}\log\frac{V_{\pi^{2}}(0)}{V_{\pi^0}(0)}.
\end{equation*}
\label{Tcounter}
\end{theorem}

\begin{remark}
Theorem \ref{Tcounter} just gives existence of a $Y_1$ that satisfies $M(s,2s)<M(s,0)$. It does not give a value of $M(s,2s)$, and not all values lead to the fundamental portfolio underperforming the market portfolio in expectation. For example, $M(s,2s)=0$ and $M(s,0)=1$ leads to outperformance in expectation by Theorem \ref{T1}. 
\end{remark}

\section{Conclusions \& Further Research}\label{conclusion}
It remains to be seen whether the assumptions and conditions of results can be satisfied and exploited in practice. Empirical validation would make this theory far more convincing. Some treatment of dividends would also help, since the results do not factor in dividends. The following paragraphs provide some ideas for empirical validation and the treatment of dividends.

The theory presented here assumes existence of a fundamental price around which the stock price reverts. So in order to take advantage of the results in practice, the fundamental price must be determined. One option is to use some combination of metrics like B/M, E/P, C/P and D/P. Another option is to try determining the fundamental price with a Bayesion method, in the spirit of a hidden Markov model; one possibility is a particle filter like in \cite{kitagawa1996monte}. 

A serious disadvantage of this theory is that it does not address dividends. If the dividends from the fundamental portfolio eventually surpass the dividends of the market portfolio, then there is no problem. Perhaps some construction of the fundamental price that uses D/P will provide this outperformance in dividends. Of course, emperical evidence will be needed to verify that the stock price is reverting around the fundamental price with the necessary strength and symmetry. 

\begin{appendices}

\section{Proofs}\label{appendix}
\subsection*{Lemma \ref{l:logV}}
\begin{proof}
The result is easily obtained from Definition \ref{d:V}.
\end{proof}

\subsection*{Lemma \ref{l:dlogV}}
\begin{proof}
The expression of $\Delta_k\log\frac{V_\pi}{V_\eta}$ is easily obtained from Lemma \ref{l:logV}.

Let $t_j,t_k\in T$ such that $t_j<t_k$. Since the $X_i$ are positive $\mathcal{A}$-measurable functions, it follows that $\frac{X_i(t_{j+1})}{X_i(t_j)}$ is $\mathcal{A}$-measurable for each $i$. The $\mathcal{A}$-measurable functions are closed under finite linear combination, so $\sum_{i=1}^n\pi_i(t_j)\frac{X_i(t_{j+1})}{X_i(t_j)}$ is $\mathcal{A}$-measurable. Using recursion on Definition \ref{d:V}, it is now clear that $\frac{V_\pi(\boldsymbol\omega,t_k)}{V_\eta(\boldsymbol\omega,t_k)}$ is a finite product of $\mathcal{A}$-measurable functions, making itself $\mathcal{A}$-measurable. Lastly, observe that $\log:(0,\infty)\to\mathbb{R}$ is Borel measurable, making $\log\frac{V_\pi(\boldsymbol\omega,t_k)}{V_\eta(\boldsymbol\omega,t_k)}$ $\mathcal{A}$-measurable.

Now for integrability. First observe that 
\begin{equation*}
\log\frac{V_\pi(\boldsymbol\omega,t_k)}{V_\eta(\boldsymbol\omega,t_k)}=\log V_\pi(\boldsymbol\omega,t_k)-\log V_\eta(\boldsymbol\omega,t_k)=\sum_{j=0}^{k-1}\Delta_j\log V_\pi-\Delta_j\log V_\eta.
\end{equation*}
Let $t_j\in T$ with $t_j<t_k$. Since integrable functions are closed under finite linear combination, it suffices to show $\Delta_j\log V_\pi$ is integrable for some arbitrary portfolio $\boldsymbol\pi$ defined on $T$. Using Definition \ref{d:V},
\begin{equation*}
\Delta_j\log V_\pi=\log\sum_{i=1}^n\pi_i(t_j)\frac{F_i(t_{j+1})}{F_i(t_j)}\exp(\Delta_jY_i).
\end{equation*}
By the concavity of $\log$, for each $\mathbf{a,x}\in(0,\infty)^n$,
\begin{equation}
\vert \log\sum_{i=1}^na_ix_i\vert \leq\max\{\sum_{i=1}^na_ix_i,\ \frac{1}{n}\sum_{i=1}^n\log (a_ix_i)\}\leq\vert \sum_{i=1}^na_ix_i\vert +\vert \sum_{i=1}^n\log (a_ix_i)\vert .
\label{e:loglinub}
\end{equation}
Using \eqref{e:loglinub}, it now suffices to show the following functions are $\mathcal{A}$-measurable and integrable:
\begin{equation}
\sum_{i=1}^n\pi_i(t_j)\frac{F_i(t_{j+1})}{F_i(t_j)}\exp(\Delta_tY_i),\quad\sum_{i=1}^n\log\Big(\pi_i(t_j)\frac{F_i(t_{j+1})}{F_i(t_j)}\Big)+\Delta_jY_i.
\label{e:dblfnc}
\end{equation}
$\mathcal{A}$-measurability follows from logic similar to what was used in showing $\log\frac{V_\pi(\boldsymbol\omega,t_k)}{V_\eta(\boldsymbol\omega,t_k)}$ is $\mathcal{A}$-measurable. 
By assumption, the expectations of $\vert Y_i(t)\vert $ and $\exp Y_i(t)$ are finite for each $i\in\{1,...,n\}$ and $t\in\mathbb{R}$. Therefore the expectations of $\vert \Delta_jY_i\vert $ and $\exp(\Delta_jY_i)$ are finite for $i=1,...,n$ and $j=0,...,k-1$. This means $\Delta_jY_i$ and $\exp(\Delta_jY_i)$ are integrable functions. Since $\pi_i$ and $F_i$ are positive, deterministic functions, it follows from linearity of the integral that both functions in \eqref{e:dblfnc} are integrable. 
\end{proof}

\begin{definition}
For each $t_{k+1}\in T$, $k\geq0$ and $m\in\{1,...,n\}$, let $f_k^{m}(Y_{m}(t_k),\ \Delta_kY_{m})=\Delta_k\log\frac{V_{\pi^{m}}}{V_{\pi^{m-1}}}$. Further, define the function
\begin{equation*}
h^{m}_k:\ \mathbb{R}^2\to\mathbb{R},\quad h^{m}_k(y,d_y)=f_k^{m}(y,d_y)+f_k^{m}(-y,-d_y).
\end{equation*}
To simplify notation, also define the function $\phi:(0,\infty)^2\to (0,\infty)^2$ such that 
\begin{equation*}
\phi(x,y)=\frac{x}{y}+\frac{y}{x}.
\end{equation*}
\end{definition}

\begin{lemma}
Let $(\Omega,d)$ be a metric space and $(\Omega,\mathcal{C})$ be a measurable space, on which two measures, $\mu$ and $\nu$, are defined. Let $h:\Omega\to\Omega$ be a continuous, injective function and $f,g:\Omega\to\mathbb{R}$ be $\mathcal{C}$-measurable functions. Let $R\in\mathcal{C}$. Suppose that for all $A\in\mathcal{C}_R$
\begin{equation*}
\mu(A)\odot\int_{h(A)}gd\nu,
\end{equation*}
If $\odot\in\{=\}$, then
\begin{equation*}
\int_Rfd\mu\odot\int_{h(R)}(f\circ h^{-1})gd\nu.
\end{equation*}
If $f$ is non-negative on $R$, then the result holds for $\odot\in\{<,\leq,=,\geq,>\}$. 
\label{l:integral}
\end{lemma}
\begin{proof}
Observe that $A\in\mathcal{C}$ implies $h(A)\in\mathcal{C}$ because $h$ is continuous and injective. The remainder of the proof is a bootstrapping argument, where the result is shown when $f$ is a characteristic function, then a non-negative step function, then a non-negative $\mathcal{C}$-measurable functions using the mean value theorem and then a $\mathcal{C}$-measurable function by taking the difference of the positive and negative parts of $f$.
\end{proof}

\subsection*{Theorem \ref{T1}}
\begin{proof}
Using a property of $\log$ and linearity of $\mathbb{E}$,
\begin{equation*}
\mathbb{E}\log\frac{V_{\pi^{m_2}}(t_{k+1})}{V_{\pi^{m_1-1}}(t_{k+1})}=\sum_{m=m_1}^{m_2}\mathbb{E}\log\frac{V_{\pi^{m}}(t_{k+1})}{V_{\pi^{m-1}}(t_{k+1})}.
\end{equation*}
Next observe that for $m\in\{m_1,...,m_2\}$, 
\begin{equation}
\log\frac{V_{\pi^{m}}(t_{k+1})}{V_{\pi^{m-1}}(t_{k+1})}=\Delta_k\log\frac{V_{\pi^{m}}}{V_{\pi^{m-1}}}+\log\frac{V_{\pi^{m}}(t_k)}{V_{\pi^{m-1}}(t_k)}.
\label{e:dklv}
\end{equation}
So by the linearity of $\mathbb{E}$, it suffices to show $\sum_{m=m_1}^{m_2}\mathbb{E}\Delta_k\log\frac{V_{\pi^{m}}}{V_{\pi^{m-1}}}\odot 0$. If $\odot\in\{\geq\}$, it is sufficient to show $\mathbb{E}\Delta_k\log\frac{V_{\pi^{m}}}{V_{\pi^{m-1}}}\geq 0$ for all $m\in\{m_1,...,m_2\}$. If $\odot\in\{>\}$, then in addition, there must be at least one $m\in\{m_1,...,m_2\}$ such that $\mathbb{E}\Delta_k\log\frac{V_{\pi^{m}}}{V_{\pi^{m-1}}}>0$ .

Suppose $\odot\in\{>\}$. Then fix $m\in\{m_1,...,m_2\}$ such that $\mu^{m}_k(R_1)>0$. From here, the goal is to show $\mathbb{E}\Delta_k\log\frac{V_{\pi^{m}}}{V_{\pi^{m-1}}}>0$. Since $\Delta_k\log\frac{V_{\pi^{m}}}{V_{\pi^{m-1}}}$ is integrable over $\boldsymbol\Omega$ (Lemma \ref{l:dlogV}) and $(\boldsymbol\Omega,\mathcal{A},\mathbb{P})$ is a product measure space, it follows that 
\begin{equation*}
\begin{split}
\mathbb{E}\Delta_k\log\frac{V_{\pi^{m}}}{V_{\pi^{m-1}}}&=\int_{\boldsymbol\Omega}\Delta_k\log\frac{V_{\pi^{m}}}{V_{\pi^{m-1}}}d\mathbb{P}\\
&=\int_{\times_{i\neq m}\Omega_i}\left[\int_{\Omega_{m}}\Delta_k\log\frac{V_{\pi^{m}}}{V_{\pi^{m-1}}}d\mathbb{P}_{m}\right]d\otimes_{i\neq m}\mathbb{P}_i\\
&=\int_{\times_{i\neq m}\Omega_i}\left[\int_{\mathbb{R}^2}f_k^{m}(y,d_y)d\mu_k^{m}(y,d_y)\right]d\otimes_{i\neq m}\mathbb{P}_i.
\end{split}
\end{equation*}
Note that $f_k^{m}$ also depends on $\times_{i\neq m}\Omega_i$. So it suffices to show $\int_{\mathbb{R}^2}f_k^{m}d\mu_k^{m}>0$ for all $\omega_i\in\Omega_i$, $i\neq m$.

Fix $\omega_i\in\Omega_i$ for all $i\neq m$. Define the regions 
\begin{equation*}
R_+=(0,\infty)\times\mathbb{R},\quad R_0=\{0\}\times\mathbb{R}.
\end{equation*}
Observe that 
\begin{equation*}
\int_{\mathbb{R}^2}f_k^{m}d\mu_k^{m}=\sum_{R\in\{R_+,R_0,-R_+\}}\int_{R}f_k^{m}d\mu_k^{m}.
\end{equation*}
Lemma \ref{l:integral} combined with $(i)$ and Lemma \ref{iii} gives 
\begin{equation*}
\int_{-R_+}f_k^{m}(y,d_y)d\mu_k^{m}(y,d_y)=\int_{R_+}f_k^{m}(-y,-d_y)d\mu_k^{m}(y,d_y).
\end{equation*}
By \eqref{dlogVbest}, which is described and justified later in this proof, it is clear that $f_k^{m}(y,d_y)=0$ whenever $y=0$. Therefore $\int_{R_0}f_k^{m}d\mu_k^{m}=0$. It follows that 
\begin{equation}
\begin{split}
\int_{\mathbb{R}^2}f_k^{m}d\mu_k^{m}&=\int_{R_+}[f_k^{m}(y,d_y)+f_k^{m}(-y,-d_y)]d\mu_k^{m}(y,d_y)\\
&=\int_{R_+}h_k^{m}d\mu_k^{m}.
\end{split}
\label{R1}
\end{equation}

From here, some useful characteristics of $h_k^{m}(y,d_y)$ are uncovered. Then it will be possible to do more with \eqref{R1}. Using Lemma \ref{l:dlogV} and substitution,
\begin{equation}
\Delta_k\log\frac{V_{\pi^{m}}}{V_{\pi^{m-1}}}=\log\frac{\sum_{j=1}^n\lambda_j^{m-1}(t_k)\sum_{i=1}^n\lambda_i^{m}(t_k)\exp\Delta_k\log X_i}{\sum_{j=1}^n\lambda_j^{m}(t_k)\sum_{i=1}^n\lambda_i^{m-1}(t_k)\exp\Delta_k\log X_i}.
\label{dlogV2}
\end{equation}
Now let 
\begin{equation*}
\begin{split}
A_k&=\sum_{i=1}^{m-1}F_i(t_k)+\sum_{i=m+1}^nX_i(t_k)\\
B_k&=\sum_{i=1}^{m-1}F_i(t_k)\exp\Delta_k\log X_i+\sum_{i=m+1}^nX_i(t_k)\exp\Delta_k\log X_i.
\end{split}
\end{equation*}
Note that $\sum_{i=1}^{m-1}\cdot=0$ if $m=1$. Now \eqref{dlogV2} can be written as
\begin{equation}
\begin{split}
\Delta_k\log\frac{V_{\pi^{m}}}{V_{\pi^{m-1}}}&=\log\frac{A_k+X_{m}(t_k)}{A_k+F_{m}(t_k)}\\
&\quad +\log\frac{B_k+F_{m}(t_k)\exp\Delta_k\log X_{m}}{B_k+X_{m}(t_k)\exp\Delta_k\log X_{m}}.
\end{split}
\label{dlogV3}
\end{equation}
Recall that $X_{m}(t_k)=F_{m}(t_k)\exp(Y_{m}(t_k))$ and $\exp\Delta_k\log F_{m}=\frac{F_{m}(t_{k+1})}{F_{m}(t_k)}$. Then \eqref{dlogV3} becomes
\begin{equation}
\begin{split}
\Delta_k\log\frac{V_{\pi^{m}}}{V_{\pi^{m-1}}}&=\log\frac{A_k+F_{m}(t_k)\exp(Y_{m}(t_k))}{A_k+F_{m}(t_k)}\\
&\quad +\log\frac{B_k+F_{m}(t_{k+1})\exp\Delta_kY_{m}}{B_k+F_{m}(t_{k+1})\exp(Y_{m}(t_k)+\Delta_kY_{m})}.
\end{split}
\label{dlogVbest}
\end{equation}
Using substitution and logorithmic properties,
\begin{equation*}
\begin{split}
h^{m}_k(y,d_y)&=\log\frac{[A_k+F_{m}(t_k)\exp y][A_k+F_{m}(t_k)\exp(-y)]}{[A_k+F_{m}(t_k)]^2}\\
&\quad +\log\frac{[B_k+F_{m}(t_{k+1})\exp d_y][B_k+F_{m}(t_{k+1})\exp(-d_y)]}{[B_k+F_{m}(t_{k+1})\exp(y+d_y)][B_k+F_{m}(t_{k+1})\exp(-y-d_y)]}.
\end{split}
\end{equation*}
Multiplication and further manipulation reveals that
\begin{equation*}
h^{m}_k(y,d_y)=\log\frac{\phi(A_k,F_{m}(t_k))+\exp y+\exp(-y)}{\phi(A_k,F_{m}(t_k))+2}+g(y,d_y),
\end{equation*}
where $g$ is defined as
\begin{equation*}
g:\ \mathbb{R}^2\to\mathbb{R},\quad g(y,d_y)=\log\frac{\phi(B_k,F_{m}(t_{k+1}))+\exp d_y+\exp(-d_y)}{\phi(B_k,F_{m}(t_{k+1}))+\exp(y+d_y)+\exp(-y-d_y)}.
\end{equation*}
Observe that for all $(y,d_y)\in\mathbb{R}^2$, $g(y,d_y)+g(y,-y-d_y)=0$. For all $y>0$, 
\begin{equation*}
\begin{split}
&\exp d_y+\exp(-d_y)>\exp(y+d_y)+\exp(-y-d_y)\\
\iff &\exp(2d_y)[1-\exp y]>\exp (-y)-1\\
\iff &\exp(2d_y)<\exp (-y)\\
\iff &d_y<-\frac{1}{2}y.
\end{split}
\end{equation*}
and similarly,
\begin{equation*}
\begin{split}
&\exp d_y+\exp(-d_y)<\exp(y+d_y)+\exp(-y-d_y)\\
\iff &d_y>-\frac{1}{2}y.
\end{split}
\end{equation*}
Therefore, for all $y>0$, 
\begin{equation*}
\begin{split}
g(y,d_y)>0\iff &d_y<-\frac{1}{2}y\\
g(y,d_y)=0\iff &d_y=-\frac{1}{2}y\\
g(y,d_y)<0\iff &d_y>-\frac{1}{2}y.
\end{split}
\end{equation*}
In addition, $\log\frac{\phi(A_k,F_{m}(t_k))+\exp y+\exp(-y)}{\phi(A_k,F_{m}(t_k))+2}>0$ for all $y>0$. This is because 
\begin{equation*}
\exp y+\exp(-y)>2\iff\frac{\exp y}{\exp y}\cdot\frac{\exp y-1}{1-\exp(-y)}>0\iff\exp y>0.\\
\end{equation*}

Now it is possible to do more with \eqref{R1}. Observe that 
\begin{equation*}
\int_{R_+}h_k^{m}d\mu_k^{m}=\sum_{R\in\{R_1\setminus R_2,R_2,R_2'\}}\int_{R}h_k^{m}d\mu_k^{m}.
\end{equation*}
Observing that $h_k^m$ is non-negative on $R_2'$, Lemma \ref{l:integral} combined with $(ii)$ and Lemma \ref{iii} gives
\begin{equation*}
\int_{R_2'}h_k^{m}(y,d_y)d\mu_k^{m}(y,d_y)\geq\int_{R_2}h_k^{m}(y,-y-d_y)d\mu_k^{m}(y,d_y).
\end{equation*}
Therefore 
\begin{equation*}
\begin{split}
\int_{R_+}h_k^{m}d\mu_k^{m}\geq&\int_{R_2}[h_k^{m}(y,d_y)+h_k^{m}(y,-y-d_y)]d\mu_k^{m}(y,d_y)\\
&+\int_{R_1\setminus R_2}h_k^{m}(y,d_y)d\mu_k^{m}(y,d_y).
\end{split}
\end{equation*}
The results of the previous paragraph make it clear that $h_k^{m}>0$ on $R_1\setminus R_2$, and
\begin{equation*}
h_k^{m}(y,d_y)+h_k^{m}(y,-y-d_y)>g(y,d_y)+g(y,-y-d_y)=0.
\end{equation*}
on $R_2$. Since $\int_{R_1}d\mu_k^{m}(y,d_y)>0$ by choice of $m$, it follows that
\begin{equation*}
\int_{R_+}h_k^{m}d\mu_k^{m}>0.
\end{equation*}

Once again, fix $m\in\{m_1,...,m_2\}$, but this time do not require $\mu_k^m(R_1)>0$. From here, the goal is to show $\mathbb{E}\Delta_k\log\frac{V_{\pi^{m}}}{V_{\pi^{m-1}}}\geq0$. By \eqref{e:dklv}, it suffices to show $\int_{\mathbb{R}^2}f_k^md\mu_k^m\geq0$. Note that the following inequality, discovered earlier in this proof, still holds.
\begin{equation}
\begin{split}
\int_{\mathbb{R}^2}f_k^md\mu_k^m\geq&\int_{R_2}[h_k^{m}(y,d_y)+h_k^{m}(y,-y-d_y)]d\mu_k^{m}(y,d_y)\\
&+\int_{R_1\setminus R_2}h_k^{m}(y,d_y)d\mu_k^{m}(y,d_y).
\end{split}
\label{e:fhhh}
\end{equation}
Moreover, the integrands on the right side of \eqref{e:fhhh} are non-negative, so their integrals are also non-negative.
\end{proof}

\begin{lemma}
Denote with $(i)_{\mathcal{B}}$ and $(ii)_{\mathcal{B}}$ the conditions $(i)$ and $(ii)$ of Theorem \ref{T1} after replacing $\mathcal{R}$ with $\mathcal{B}$. Then $(i)$ implies $(i)_{\mathcal{B}}$ and $(ii)$ implies $(ii)_{\mathcal{B}}$.
\label{iii}
\end{lemma}
\begin{proof}
Fix $m\in\{1,...,n\}$ and $k\geq0$ such that $t_{k+1}\in T$. Suppose condition $(i)$ of Theorem \ref{T1} holds. Let $B\in\mathcal{B}$. Since $\mu_k^i$ is the unique extension to $\mathcal{B}$ of the measure $\mu_k^m\vert _{\mathcal{R}}$, it follows that
\begin{equation*}
\begin{split}
\mu_k^m(B)&=\inf\Bigg\{\sum_{j=1}^J\mu_k^m(R_j):\ R_j\in\mathcal{R}\ \forall j,\ B\subset\cup_{j=1}^JR_j\Bigg\}\\
&=\inf\Bigg\{\sum_{j=1}^J\mu_k^m(-R_j):\ R_j\in\mathcal{R}\ \forall j,\ B\subset\cup_{j=1}^JR_j\Bigg\}\\
&=\inf\Bigg\{\sum_{j=1}^J\mu_k^m(R_j):\ R_j\in\mathcal{R}\ \forall j,\ -B\subset\cup_{j=1}^JR_j\Bigg\}\\
&=\mu_k^m(-B).
\end{split}
\end{equation*}
Therefore $(i)_{\mathcal{B}}$ is satisfied. 

Suppose condition $(ii)$ of Theorem \ref{T1} holds. Let $B\in\mathcal{B}_{R_2}$. Given $A\in\mathcal{B}$, observe that $\mu_k^i\vert _{\mathcal{B}_{A}}$ is the unique extension to $\mathcal{B}_{A}$ of the measure $\mu_k^m\vert _{\mathcal{R}_{A}}$. It follows that
\begin{equation*}
\begin{split}
\mu_k^m(B)&=\inf\Bigg\{\sum_{j=1}^J\mu_k^m(R_j):\ R_j\in\mathcal{R}_{R_2}\ \forall j,\ B\subset\cup_{j=1}^JR_j\Bigg\}\\
&\leq\inf\Bigg\{\sum_{j=1}^J\mu_k^m(R_j'):\ R_j\in\mathcal{R}_{R_2}\ \forall j,\ B\subset\cup_{j=1}^JR_j\Bigg\}\\
&=\inf\Bigg\{\sum_{j=1}^J\mu_k^m(R_j):\ R_j\in\mathcal{R}_{R_2'}\ \forall j,\ B'\subset\cup_{j=1}^JR_j\Bigg\}\\
&=\mu_k^m(B').
\end{split}
\end{equation*}
Therefore $(ii)_{\mathcal{B}}$ is satisfied. 
\end{proof}

\begin{lemma}
Fix $m\in\{1,...,n\}$. Consider the following conditions.
\begin{enumerate}[(a)]
\item $\nu^{m}_k(y,R)=\nu^{m}_k(-y,-R)$ for all $R\in\mathcal{R}$, $u_k^m$-almost everywhere.
\item $\nu^{m}_k(y,R)\leq \nu^{m}_k(y,R')$ for all $R\in\mathcal{R}_{R_2}$, $u_k^m$-almost everywhere.
\end{enumerate}
Suppose $\mathbb{P}(Y_m(t_0)\in I)=\mathbb{P}(Y_m(t_0)\in\{-y:\ y\in I\})$ for all intervals $I\subset\mathbb{R}$. Recall conditions (i) and (ii) of Theorem \ref{T1}. If (a) holds for all $k\leq K$, then so does (i). If (b) holds for all $k\leq K$, then so does (ii).
\label{l:conditional}
\end{lemma}
\begin{proof}
First observe that if $u_k^m(I)=u_k^m(\{-z:\ z\in I\})$ for all intervals $I\subset\mathbb{R}$, uniqueness of the extension to a measure on the Borel sigma algebra implies the equality also holds when $I$ is a Borel subset of $\mathbb{R}$. 

Fix the interval $I\subset\mathbb{R}$, and suppose $(a)$ holds for all $k\leq K$. An inductive argument shows that (i) holds and $\mathbb{P}(Y_m(t_k)\in I)=\mathbb{P}(Y_m(t_k)\in\{-y:\ y\in I\})$ for all $k\leq K$. 

By assumption, $u_0^m(I)=u_0^m(\{-z:\ z\in I\})$. Now suppose $u_k^m(I)=u_k^m(\{-z:\ z\in I\})$. Let $I_1\times I_2=R\in\mathcal{R}$. Then by Lemma \ref{l:integral}, $(a)$ and conditional probability theory, 
\begin{equation*}
\begin{split}
\mu^m_k(R)&=\int_{I_1}\nu^m_k(y,R)du^m_k(y)\\
&=\int_{I_1}\nu^m_k(-y,-R)du^m_k(y)\\
&=\int_{\{-y:\ y\in I_1\}}\nu^{m}_k(y,-R)du^m_k(y)\\
&=\mu^{m}_k(-R).
\end{split}
\end{equation*}
Therefore $(i)$ holds, and by Lemma \ref{iii}, $(i)_{\mathcal{B}}$ also holds. Moreover, if $k+1\leq K$, then
\begin{equation*}
\begin{split}
u_{k+1}^m(I)&=\mu^m_k(\{(y,d_y)\in\mathbb{R}^2:\ y+d_y\in I\})\\
&=\mu^m_k(-\{(y,d_y)\in\mathbb{R}^2:\ y+d_y\in I\})\\
&=u_{k+1}^m\{-y:\ y\in I\}).
\end{split}
\end{equation*}
The induction is complete.

Now fiix $k\leq K$ and suppose $(b)$ holds. Let $R\in\mathcal{R}_{R_2}$ and set $I=\{y:\ (y,d_y)\in R\}$. Then by conditional probability theory, 
\begin{equation*}
\begin{split}
\mu^m_k(R)&=\int_{I}\nu^m_k(y,R)du^m_k(y)\\
&\leq\int_{I}\nu^{m}_k(y,R')du^m_k(y)\\
&=\mu^{m}_k(R').
\end{split}
\end{equation*}
Therefore $(ii)$ holds.
\end{proof}

\subsection*{Theorem \ref{T3}}
\begin{proof}
By Theorem \ref{T1} and Lemma \ref{l:conditional}, it suffices to show $\sum_{m=m_1}^{m_2}\mu^{m}_K(R_1)>0$. Using the assumption, fix $m\in\{m_1,...,m_2\}$ such that 
\begin{equation*}
\begin{split}
&\nu_{K-1}^m(y,\{(z,-z):\ z\in\mathbb{R}\})<1\ \ u_{K-1}^m\text{-almost everywhere},\\
&\nu_K^m(y,R_1)+\mathcal{X}_{(-\infty,0]}(y)>0\ \ u_K^m\text{-almost everywhere}.
\end{split}
\end{equation*} 
By conditional probability theory,
\begin{equation*}
u_K^m(\{0\})=\int_\mathbb{R}\nu_{K-1}^m(y,\{(z,-z):\ z\in\mathbb{R}\})du_{K-1}^m(y)<1.
\end{equation*}
By Lemma \ref{l:conditional}, $(i)$ from Theorem \ref{T1} is satisfied for $K$. Therefore 
\begin{equation*}
\begin{split}
0<u_K^m(\mathbb{R}\setminus\{0\})&=\mu_K^m((-\infty,0)\times\mathbb{R})+\mu_K^m((0,\infty)\times\mathbb{R})\\
&=2\mu_K^m((0,\infty)\times\mathbb{R})=2u_K^m((0,\infty)).
\end{split}
\end{equation*}
Since $u_K^m((0,\infty))>0$, another application of conditional probability theory gives
\begin{equation*}
\mu_K^m(R_1)=\int_{(0,\infty)}\nu_K^m(y,R_1)du_{K}^m(y)>0.
\end{equation*}
\end{proof}

\subsection*{Corollary \ref{c:ou}}
\begin{proof}
It suffices to show all conditions of Theorem \ref{T3} are satisfied. $(c)$ is satisfied because for all $m\in\{m_1,...,m_2\}$, $Y_m(0)$ is a non-trivial random variable with symmetric distribution about $0$. Now to show that $(a)$ and $(b)$ hold for all $k\leq K$.

Fix $m\in\{m_1,...,m_2\}$, $k\leq K$ and $y\in\mathbb{R}$. For $R\in\mathcal{R}\cup\mathcal{R}_{R_2}$, the Ornstein-Uhlenbeck process $Y_m(t)$ has conditional probability measure $\nu_k^m(y,R)=\int_{R_y}p_y(x)dx$, where $R_y=\{x:\ (y,x)\in R\}$ and $p_y$ is the density for a Normal random variable with mean 
\begin{equation*}
y\cdot(\exp(-\theta_m(t_{k+1}-t_k))-1),
\end{equation*} 
and variance that does not depend on $y$. Since the mean is an odd function in $y$, and the variance does not depend on $y$, $p_y(x)=p_{-y}(-x)$ for all $x,y\in\mathbb{R}$. It follows that $(a)$ holds. Next observe that $t_{k+1}-t_k\geq\frac{\ln 2}{\theta_m}$ implies the mean associated with $p_y$ is no greater than $-\frac{1}{2}y$. Therefore $p_y(x)\leq p_y(-y-x)$ for all $y>0$ and $x>-\frac{1}{2}y$. It follows that $(b)$ holds.

The remaining two conditions of Theorem \ref{T3} hold because $p_y$ is the density for a continuous random variable on $\mathbb{R}$. In particular, 
\begin{equation*}
\begin{split}
&\nu_{K-1}^m(y,\{(z,-z):\ z\in\mathbb{R}\})=\int_{\{-y\}}p_y(x)dx=0,\\
&\nu_{K}^m(y,R_1\})=\int_{(R_1)_y}p_y(x)dx>0.
\end{split}
\end{equation*}
\end{proof}

\subsection*{Corollary \ref{corar}}
\begin{proof}
It suffices to show all conditions of Theorem \ref{T3} are satisfied. $(c)$ is satisfied because for all $m\in\{m_1,...,m_2\}$, $Y_m(0)$ is a non-trivial random variable with symmetric distribution about $0$. Now to show that $(a)$ and $(b)$ hold for all $k\geq 0$.

Fix $m\in\{m_1,...,m_2\}$ and $k\geq 0$. Observe that $\Delta_kY_m=(\theta_m-1)Y_m(k)+Z_m(k)$. Therefore 
\begin{equation*}
(\Delta_kY_m\vert \ Y_m(k)=y)=(\theta_m-1)y+Z_m(k+1).
\end{equation*}
Since $Z_m(k+1)$ is symmetric with mean 0,
\begin{equation*}
(\Delta_kY_m\vert \ Y_m(k)=y)\stackrel{\text{d}}{=}-(\Delta_kY_m\vert \ Y_m(k)=-y).
\end{equation*}
So $(a)$ holds. Next observe that $\theta_m\leq\frac{1}{2}$ implies the mean of $(\Delta_kY_m\vert \ Y_m(k)=y)$ is no greater than $-\frac{1}{2}y$. Moreover, $(\Delta_kY_m\vert \ Y_m(k)=y)$ has symmetric distribution about its mean. So $(b)$ is satisfied.

Now fix $k>0$. From here, the goal is to show the remaining two conditions of Theorem \ref{T3} are satisfied. Since $\mathbb{P}(Z_m(k)>a)>0$ for all $a\in\mathbb{R}$, the following holds for all $y\in\mathbb{R}$.
\begin{equation*}
\begin{split}
\nu_{k-1}^m(y,\{(z,-z):\ z\in\mathbb{R}\})&=\mathbb{P}((\theta_m-1)y+Z_m(k)=-y)\\
&\leq1-\mathbb{P}(Z_m(k)>-\theta_my)\\
&<1.
\end{split}
\end{equation*}
Similarly, for all $y\in(0,\infty)$,
\begin{equation*}
\nu_{k}^m(y,R_1)=\mathbb{P}\big(Z_m(k+1)\geq-\frac{1}{2}y-(\theta_m-1)y\big)>0.
\end{equation*}
\end{proof}

\subsection*{Corollary \ref{corMarkov}}
\begin{proof}
It suffices to show all conditions of Theorem \ref{T3} are satisfied. $(c)$ is satisfied because for all $m\in\{m_1,...,m_2\}$, $Y_m(0)$ is a non-trivial random variable with symmetric distribution about $0$. Now to show that $(a)$ and $(b)$ hold for all $k\geq 0$.

Fix $m\in\{m_1,...,m_2\}$ and $k\geq 0$. Assume that $k_1$ and $k_2$ denote integers. Given $B\in\mathcal{B}$, the Markov chain $Y_m$ has conditional probability measure 
\begin{equation*}
\nu_k^m(k_1s,B)=\sum_{k_2s\in B_{k_1s}}P(k_1s,(k_1+k_2)s),
\end{equation*} 
where $B_y=\{x:\ (y,x)\in B\}$. Let $R\in\mathcal{R}$. Using assumption $(i)$,
\begin{equation*}
\begin{split}
\nu_k^m(k_1s,R)&=\sum_{k_2s\in R_{k_1s}}P(k_1s,k_2s)\\
&=\sum_{k_2s\in R_{k_1s}}P(-k_1s,-k_2s)\\
&=\nu_k^m(-k_1s,-R).
\end{split}
\end{equation*}
Therefore $(a)$ is satisfied. Now let $R\in\mathcal{R}_{R_2}$. Using assumption $(ii)$, 
\begin{equation*}
\begin{split}
\nu_k^m(k_1s,R)&=\sum_{k_2s\in R_{k_1s}}P(k_1s,(k_1+k_2)s)\\
&\leq\sum_{k_2s\in R_{k_1s}}P_m(k_1s,-k_2s)\\
&=\nu_k^m(k_1s,R').
\end{split}
\end{equation*}
Therefore $(b)$ is satisfied.

Now fix $k>0$. From here, the goal is to show the remaining two conditions of Theorem \ref{T3} are satisfied. Using assumption $(iv)$ and then $(v)$,
\begin{equation*}
\begin{split}
&\nu_{K-1}^m(y,\{(z,-z):\ z\in\mathbb{R}\})=P_m(y,0)<1\ \text{for all } y\in S,\\
&\nu_K^m(y,R_1)+\mathcal{X}_{(-\infty,0]}(y)\\
&\quad=\sum_{d_y\geq -\frac{1}{2}y,\ d_y\in S}P_m(y,y+d_y)\mathcal{X}_{(0,\infty)}(y)+\mathcal{X}_{(-\infty,0]}(y)>0\ \text{for all } y\in S.
\end{split}
\end{equation*} 
\end{proof}

\begin{lemma}
Let $(\mathbb{R}^2,\mathcal{A},\mu)$ be a measure space, $f:\mathbb{R}^2\to[0,\infty)$ be an $\mathcal{A}$-measurable function, and $R\in\mathcal{A}$ such that $R\cap R'=\emptyset$. Suppose $r:\mathbb{R}^2\to[0,1)$ is an $\mathcal{A}$-measurable function such that for all $A\in\mathcal{A}_{R'}$, 
\begin{equation*}
\mu(A')\geq\int_{A}\frac{r}{1-r}d\mu,
\end{equation*}
Then
\begin{equation*}
\int_{R\cup R'}fd\mu\geq\int_R\frac{1}{1-r(x,y)}[(1-r(x,y))f(x,y)+r(x,y)f(x,-x-y)]d\mu(x,y).
\end{equation*}
\label{lself}
\end{lemma}
\begin{proof}
Applying disjointness of $R$ and $R'$ and then Lemma \ref{l:integral},
\begin{equation*}
\begin{split}
\int_{R\cup R'}fd\mu&=\int_{R}fd\mu+\int_{R'}fd\mu\\
&\geq\int_{R}fd\mu+\int_Rf(x,-x-y)\frac{r(x,y)}{1-r(x,y)}d\mu(x,y)\\
&=\int_{R}f(x,y)+f(x,-x-y)\frac{r(x,y)}{1-r(x,y)}d\mu(x,y).
\end{split}
\end{equation*}
\end{proof}

\subsection*{Theorem \ref{T2}}
\begin{proof}
By assumption of the Theorem, fix $m\in\{m_1,...,m_2\}$ and $k<K$ such that 
\begin{equation*}
\mu^{m}_k(\{(y,d_y)\in\mathbb{R}^2:\ y>0,\  d_y\geq -\frac{1}{2}y\})>0.
\end{equation*}
Using the same logic and notation as in the proof of Theorem \ref{T1}, it suffices to show 
\begin{equation*}
\sum_{R\in\{R_1\setminus R_2,R_2,R_2'\}}\int_{R}h_k^{m}d\mu_k^{m}>0.
\end{equation*}
By Lemma \ref{lself} combined with $(ii)$ and Lemma \ref{l:T2ii},
\begin{equation*}
\begin{split}
&\sum_{R\in\{R_2,R_2'\}}\int_{R}h_k^{m}d\mu_k^{m}\\
&\quad\geq\int_{R_2}\frac{1}{1-r_k^m(y,d_y)}[(1-r_k^m(y,d_y))h_k^{m}(y,d_y)+r_k^m(y,d_y)h_k^{m}(y,-y-d_y)]d\mu_k^{m}(y,d_y).
\end{split}
\end{equation*}
Since $\int_{R_1}d\mu_k^{m}(y,d_y)>0$, $h_k^m(y,d_y)>0$ on $R_1\setminus R_2$ and $\frac{1}{1-r_k^m(y,d_y)}\geq 1$ on $R_2$, it now suffices to show 
\begin{equation}
(1-r_k^m(y,d_y))h_k^m(y,d_y)+r_k^m(y,d_y)h_k^m(y,-y-d_y)>0,\quad\forall (y,d_y)\in R_2.
\label{mrh}
\end{equation}
Fix $(y,d_y)\in R_2$. From the proof of Theorem \ref{T1}, reuse the notation
\begin{equation*}
h_k^m(y,d_y)=\log\frac{\phi(A_k,F_m(t_k))+\exp y+\exp(-y)}{\phi(A_k,F_m(t_k))+2}+g(y,d_y).
\end{equation*}
By Theorem \ref{T1}, \eqref{mrh} is satisfied if $r_k^m(y,d_y)\geq .5$. Now suppose $r_k^m(y,d_y)<.5$. Using logic from the proof of Theorem \ref{T1} and the fact that $1-r_k^m(y,d_y)>r_k^m(y,d_y)$,
\begin{equation*}
\begin{split}
&(1-r_k^m(y,d_y))h_k^m(y,d_y)+r_k^m(y,d_y)h_k^m(y,-y-d_y)\\
&\quad =\log\frac{\phi(A_k,F_m(t_k))+\exp y+\exp(-y)}{\phi(A_k,F_m(t_k))+2}+[1-2r_k^m(y,d_y)]g(y,d_y).
\end{split}
\end{equation*}

Conditions $(iii)$ and $(iv)$ imply 
\begin{equation*}
\begin{split}
&F_{m}(t_k)\exp(-\delta_1)\leq F_{m}(t_{k+1})\leq F_{m}(t_k)\exp\delta_1,\\
&A_k\exp(-\delta_1-\delta_2)\leq B_k\leq A_k\exp(\delta_1+\delta_2).
\end{split}
\end{equation*}
It follows that
\begin{equation*}
\phi(B_k,F_{m}(t_{k+1}))\exp(-2\delta_1-\delta_2)\leq \phi(A_k,F_m(t_k))\leq \phi(B_k,F_{m}(t_{k+1}))\exp(2\delta_1+\delta_2).
\end{equation*}
Set $\delta=2\delta_1+\delta_2$.

Differentiating $\log\frac{x+\exp y+\exp(-y)}{x+2}$ with respect to $x$ shows it to be decreasing as $x$ increases. Therefore
\begin{equation*}
\log\frac{\phi(A_k,F_m(t_k))+\exp y+\exp(-y)}{\phi(A_k,F_m(t_k))+2}\geq\log\frac{\phi(B_k,F_{m}(t_{k+1}))\exp(\delta)+\exp y+\exp(-y)}{\phi(B_k,F_{m}(t_{k+1}))\exp(\delta)+2}.
\end{equation*}
So now it suffices to show
\begin{equation}
\begin{split}
&\log\frac{x\exp(\delta)+\exp y+\exp(-y)}{x\exp(\delta)+2}\\
&+[1-2r_k^m(y,d_y)]\log\frac{x+\exp d_y+\exp(-d_y)}{x+\exp(y+d_y)+\exp(-y-d_y)}>0,
\end{split}
\label{dk}
\end{equation}
where $x=\phi(B_k,F_{m}(t_{k+1}))$. Some manipulation reveals that \eqref{dk} is equivalent to 
\begin{equation}
\frac{1}{2}\left[1-\frac{\log\frac{x\exp(\delta)+\exp y+\exp(-y)}{x\exp(\delta)+2}}{\log\frac{x+\exp(y+d_y)+\exp(-y-d_y)}{x+\exp d_y+\exp(-d_y)}}\right]<r_k^m(y,d_y).
\label{dk2}
\end{equation}

Taking the derivative of the left side of \eqref{dk} with respect to $x$ gives
\begin{equation*}
\begin{split}
&\frac{1}{\exp\delta}\left[\frac{2-\exp y-\exp(-y)}{(x+\frac{\exp y+\exp(-y)}{\exp\delta})(x+\frac{2}{\exp\delta})}\right]\\
&+(1-2r_k^m(y,d_y))\left[\frac{\exp(y+d_y)+\exp(-y-d_y)-\exp d_y-\exp(-d_y)}{(x+\exp(y+d_y)+\exp(-y-d_y))(x+\exp d_y+\exp(-d_y))}\right].
\end{split}
\end{equation*}
By $(iv)$ and the definition of $\delta$, 
\begin{equation*}
\begin{split}
\frac{\exp y+\exp(-y)}{\exp\delta}&\leq \exp(y+d_y)+\exp(-y-d_y),\\
\frac{2}{\exp\delta}&\leq \exp d_y+\exp(-d_y).
\end{split}
\end{equation*}
So the derivative of the left side of \eqref{dk} is non-positive provided
\begin{equation*}
\begin{split}
&\frac{2-\exp y-\exp(-y)}{\exp\delta}\\
&+(1-2r_k^m(y,d_y))[\exp(y+d_y)+\exp(-y-d_y)-\exp d_y-\exp(-d_y)]\leq0,
\end{split}
\end{equation*}
or equivalently, when
\begin{equation}
\frac{1}{2}\left[1-\exp(-\delta)\frac{\exp y+\exp(-y)-2}{\exp(y+d_y)+\exp(-y-d_y)-\exp d_y-\exp(-d_y)}\right]\leq r(y,d_y).
\label{dk4}
\end{equation}
Therefore, given \eqref{dk4}, the minimum of the left side of \eqref{dk} occurs at $x\to\infty$. So \eqref{dk2} is satisfied for all $x$ provided
\begin{equation}
\lim_{x\to\infty}\frac{1}{2}\left[1-\frac{\log\frac{x\exp(\delta)+\exp y+\exp(-y)}{x\exp(\delta)+2}}{\log\frac{x+\exp(y+d_y)+\exp(-y-d_y)}{x+\exp d_y+\exp(-d_y)}}\right]<r(y,d_y).
\label{dk3}
\end{equation}
Two applications of L'Hopital's rule reveals that the left side of \eqref{dk3} is the same as the left side of \eqref{dk4}.
\end{proof}

\begin{lemma}
Denote with $(ii)_{\mathcal{B}}$ the condition $(ii)$ of Theorem \ref{T2} after replacing $\mathcal{R}$ with $\mathcal{B}$. Then $(ii)$ implies $(ii)_{\mathcal{B}}$.
\label{l:T2ii}
\end{lemma}
\begin{proof}
Fix $m\in\{1,...,n\}$ and $k\geq0$ such that $t_{k+1}\in T$. Suppose condition $(ii)$ of Theorem \ref{T1} holds. Let $B\in\mathcal{B}_{R_2}$. Given $A\in\mathcal{B}$, observe that $\mu_k^i\vert _{\mathcal{B}_{A}}$ is the unique extension to $\mathcal{B}_{A}$ of the measure $\mu_k^m\vert _{\mathcal{R}_{A}}$. It follows that
\begin{equation*}
\begin{split}
\mu_k^m(B')&=\inf\Bigg\{\sum_{j=1}^J\mu_k^m(R_j):\ R_j\in\mathcal{R}_{R_2'}\ \forall j,\ B'\subset\cup_{j=1}^JR_j\Bigg\}\\
&=\inf\Bigg\{\sum_{j=1}^J\mu_k^m(R_j'):\ R_j\in\mathcal{R}_{R_2}\ \forall j,\ B\subset\cup_{j=1}^JR_j\Bigg\}\\
&\geq\inf\Bigg\{\sum_{j=1}^J\int_{R_j}\frac{r_k^m}{1-r_k^m}d\mu_k^m:\ R_j\in\mathcal{R}_{R_2}\ \forall j,\ B\subset\cup_{j=1}^JR_j\Bigg\}\\
&=\inf\Bigg\{\int_{\cup_{j=1}^JR_j}\frac{r_k^m}{1-r_k^m}d\mu_k^m:\ R_j\in\mathcal{R}_{R_2}\ \forall j,\ B\subset\cup_{j=1}^JR_j\Bigg\}\\
&\geq\int_{B}\frac{r_k^m}{1-r_k^m}d\mu_k^m.
\end{split}
\end{equation*}
Therefore $(ii)_{\mathcal{B}}$ holds. 
\end{proof}

\subsection*{Theorem \ref{Tcounter}}
\begin{proof}
Require $M(s,(1\pm1)s)=M(-s,-(1\pm1)s)$. A straightforward modification to the proof of Theorem \ref{T1} reveals that it is sufficient to show 
\begin{equation}
\int_{R_+}h_0^1d\mu_0^1<0.
\label{countere1}
\end{equation}
Using logic from the proof of Corollary \ref{corMarkov}, \eqref{countere1} is equivalent to 
\begin{equation}
\begin{split}
\mathbb{P}(Y_1(0)=s)[h_0^1(s,s)M(s,2s)+h_0^1(s,-s)M(s,0)]<0.
\end{split}
\label{countere2}
\end{equation}
Require $\mathbb{P}(Y_1(0)=s)=\frac{1}{2}$ and $M(s,2s)+M(s,0)=1$. Then \eqref{countere2} holds if and only if
\begin{equation}
\begin{split}
M(s,2s)[h_0^1(s,s)-h_0^1(s,-s)]+h_0^1(s,-s)<0.
\end{split}
\label{countere33}
\end{equation}
Using the definition of $h_0^1$ and requiring $A>0$, the left side of \eqref{countere33} reduces to
\begin{equation*}
M(s,2s)\log\frac{\phi(1,A)+2}{\phi(1,A)+\exp(2s)+\exp(-2s)}+2\log\frac{\phi(1,A)+\exp s+\exp(-s)}{\phi(1,A)+2}.
\end{equation*}
Therefore \eqref{countere33} is equivalent to
\begin{equation}
\frac{2\log\frac{\phi(1,A)+\exp s+\exp(-s)}{\phi(1,A)+2}}{\log\frac{\phi(1,A)+2}{\phi(1,A)+\exp(2s)+\exp(-2s)}}<M(s,2s).
\label{countere44}
\end{equation}
Two applications of L'Hopital's rule and some manipulation reveals that 
\begin{equation*}
\begin{split}
\lim_{x\to\infty}\frac{2\log\frac{x+\exp s+\exp(-s)}{x+2}}{\log\frac{x+2}{x+\exp(2s)+\exp(-2s)}}&=2\cdot\frac{\exp s+\exp(-s)-2}{\exp(2s)+\exp(-2s)-2}\\
&=2\cdot\frac{(\exp s-1)(1-\exp(-s))}{(\exp s+1)(\exp s-1)(1+\exp(-s))(1-\exp(-s))}\\
&=\frac{2}{(\exp s+1)(1+\exp(-s))}\\
&<\frac{1}{2}.
\end{split}
\end{equation*}
Let $r=\frac{2}{(\exp s+1)(1+\exp(-s))}$. Set $M(s,2s)=\frac{1}{2}(r+\frac{1}{2})$, and choose $A$ large enough so that \eqref{countere44} holds.
\end{proof}

\end{appendices}


\bibliography{sn-bibliography}


\end{document}